\documentclass[11pt]{article}

\usepackage[active]{srcltx}
\usepackage{hyperref}
\usepackage{amsmath,amssymb}
\usepackage[all]{xy}

 1 
\newtheorem{theorem}{Theorem}
\newtheorem{proposition}{Proposition}

\newtheorem{definition}{Definition}

\def\qed{\ifvmode\Realemovelastskip\fi
{\unskip\nobreak\hfil\penalty50\hbox{}\nobreak\hfil \hbox{\vrule
height1.2ex width1.2ex}\parfillskip=0pt \finalhyphendemerits=0
\par\smallskip}}
\def\qedr{\ifvmode\Realemovelastskip\fi
{\unskip\nobreak\hfil\penalty50\hbox{}\nobreak\hfil \hbox{
$\diamond$}\parfillskip=0pt \finalhyphendemerits=0
\par\smallskip}}

\newenvironment{proof}{\noindent{\sl Proof:~~~}}{\quad \qed}

\def\beq{\begin{equation}}
\def\eeq{\end{equation}}
\def\bea{\begin{eqnarray}}
\def\eea{\end{eqnarray}}
\def\beann{\begin{eqnarray*}}
\def\eeann{\end{eqnarray*}}
\def\beasn{\begin{sneqnarray}}
\def\eeasn{\end{sneqnarray}}
\def\ben{\begin{enumerate}}
\def\een{\end{enumerate}}
\def\bit{\begin{itemize}}
\def\eit{\end{itemize}}

\def\derpar#1#2{\displaystyle\frac{\partial{#1}}{\partial{#2}}}

\def\restric#1#2{\left.#1\right|_{#2}}

\newcommand{\W}{\mathcal{W}}

\renewcommand{\P}{\mathcal{P}}
\newcommand{\vf}{\mathfrak{X}}
\newcommand{\df}{\mathit{\Omega}}
\newcommand{\Lag}{\mathcal{L}}
\newcommand{\Leg}{\mathcal{FL}}
\renewcommand{\d}{\textnormal{d}}
\newcommand{\Nat}{\mathbb{N}}

\newcommand{\R}{\mathbb{R}}

\def\Tan{{\rm T}}
\def\Lie{\mathop{\rm L}\nolimits}
\def\inn{\mathop{i}\nolimits}

\def\tabaddress#1{{\small\it\begin{tabular}[t]{c}#1
\\[1.2ex]\end{tabular}}}
\def\qed{\ifvmode\removelastskip\fi
{\unskip\nobreak\hfil\penalty50\hbox{}\nobreak\hfil \hbox{\vrule
height1.2ex width1.2ex}\parfillskip=0pt \finalhyphendemerits=0
\par\smallskip}}

\title{VARIATIONAL PRINCIPLES FOR MULTISYMPLECTIC SECOND-ORDER CLASSICAL FIELD THEORIES}

\author{
{\sc  Pedro Daniel Prieto-Mart\'\i nez\thanks{{\bf e}-{\it mail}:
   peredaniel@ma4.upc.edu} }\\
   {\sc Narciso Rom\'an-Roy\thanks{{\bf e}-{\it mail}:
   nrr@ma4.upc.edu}}  \\
   \tabaddress{Departamento de Matem\'atica Aplicada IV.
   Edificio C-3, Campus Norte UPC\\
   C/ Jordi Girona 1. 08034 Barcelona. Spain}
}
   \date{\today}

\begin{document}

\maketitle

\pagestyle{myheadings}

\thispagestyle{empty}

\begin{abstract}
We state a unified geometrical version of the variational principles for second-order classical
field theories. The standard Lagrangian and Hamiltonian variational principles and the
corresponding field equations are recovered from this unified framework.
\end{abstract}


\vbox{\raggedleft AMS s.\,c.\,(2010): 70S05, 49S05, 70H50}\null 
\markright{\rm P.D. Prieto-Mart\'{\i}nez, N. Rom\'{a}n-Roy:   \sl Variational principles for second-order CFT}

\clearpage

\tableofcontents

\vspace{15pt}

\section{Introduction}

As stated in \cite{art:Vitagliano13}, the field equations of a classical field theory arising from
a partial differential Hamiltonian system (in the sense of \cite{art:Vitagliano13}) are locally variational,
that is, they can be derived using a variational
principle. In this work we use the geometric Lagrangian-Hamiltonian formulation for second-order
classical field theories given in \cite{art:Prieto_Roman14} to state the variational principles for
this kind of theories from a geometric point of view, thus giving a different point of view and
completing previous works on higher-order classical field theories
\cite{art:Campos_DeLeon_Martin_Vankerschaver09,art:Vitagliano10}.

\vspace{0.25\baselineskip}

\noindent(All the manifolds are real, second countable and $C^\infty$. The maps and the structures are assumed
to be $C^\infty$. Usual multi-index notation introduced in \cite{book:Saunders89} is used).

\section{Higher-order jet bundles}
\label{geomset}

(See \cite{book:Saunders89} for details).
Let $M$ be an orientable $m$-dimensional smooth manifold, and let $\eta \in \df^m(M)$ be a volume
form for $M$. Let $E \stackrel{\pi}{\longrightarrow} M$ be a bundle with $\dim E = m + n$. If
$k \in \Nat$, the \textsl{$k$th-order jet bundle} of the projection $\pi$, $J^k\pi$, is the manifold
of the $k$-jets of local sections $\phi \in \Gamma(\pi)$; that is, equivalence classes of local
sections of $\pi$ by the relation of equality on every partial derivative up to order $k$. A point
in $J^k\pi$ is denoted by $j^k_x\phi$, where $x \in M$ and $\phi \in \Gamma(\pi)$ is a representative
of the equivalence class. We have the following natural projections: if $r \leqslant k$,
$$
\begin{array}{rcl}
\pi^k_r \colon J^k\pi & \longrightarrow & J^r\pi \\
j^k_x\phi & \longmapsto & j^r_x\phi
\end{array}
\quad \ \quad
\begin{array}{rcl}
\pi^k \colon J^k\pi & \longrightarrow & E \\
j^k_x\phi & \longmapsto & \phi(x)
\end{array}
\quad \ \quad
\begin{array}{rcl}
\bar{\pi}^k \colon J^k\pi & \longrightarrow & M \\
j^k_x\phi & \longmapsto & x
\end{array}
$$
Observe that $\pi^s_r\circ\pi^k_s = \pi^k_r$, $\pi^k_0 = \pi^k$, $\pi^k_k = \textnormal{Id}_{J^k\pi}$,
and $\bar{\pi}^k = \pi \circ \pi^k$.

If local coordinates in $E$ adapted to the bundle structure are $(x^i,u^\alpha)$,
$1 \leqslant i \leqslant m$, $1 \leqslant \alpha \leqslant n$, then local coordinates in $J^k\pi$ are
denoted $(x^i,u_I^\alpha)$, with $0 \leqslant |I| \leqslant k$.

If $\psi \in \Gamma(\pi)$, we denote the \textsl{$k$th prolongation} of $\phi$ to $J^k\pi$ by
$j^k\phi \in \Gamma(\bar{\pi}^k)$.

\begin{definition}
A section $\psi \in \Gamma(\bar{\pi}^{k})$ is \textnormal{holonomic} if
$j^k(\pi^{k} \circ \psi) = \psi$; that is, $\psi$ is the $k$th prolongation of a section
$\phi = \pi^{k} \circ \psi \in \Gamma(\pi)$.
\end{definition}

In the following we restrict ourselves to the case $k=2$. According to \cite{art:Saunders_Crampin90},
consider the subbundle of fiber-affine maps $J^1\bar{\pi}^1 \to \R$ which are constant on the
fibers of the affine subbundle $(\bar{\pi}^1)^*(\Lambda^2\Tan^*M) \otimes (\pi^1)^*(V\pi)$ of
$J^1\bar{\pi}^1$ over $J^1\pi$. This subbundle is canonically diffeomorphic to the
$\pi_{J^1\pi}$-transverse submanifold $J^2\pi^\dagger$ of $\Lambda_2^m(J^1\pi)$ defined locally
by the constraints $p_\alpha^{ij} = p_\alpha^{ji}$, which fibers over $J^{1}\pi$ and $M$ with
projections
$\pi_{J^1\pi}^\dagger \colon J^2\pi^\dagger \to J^1\pi$ and $\bar{\pi}_{J^1\pi}^\dagger \colon J^2\pi^\dagger \to M$,
respectively. The submanifold $j_s \colon J^2\pi^\dagger \hookrightarrow \Lambda_2^m(J^1\pi)$ is the
\textsl{extended $2$-symmetric multimomentum bundle}.

All the canonical geometric structures in $\Lambda_2^m(J^1\pi)$ restrict to
$J^2\pi^\dagger$. Denote $\Theta_1^s = j_s^*\Theta_1 \in \df^{m}(J^2\pi^\dagger)$
and $\Omega_1^s = j_s^*\Omega_1 \in \df^{m+1}(J^2\pi^\dagger)$ the pull-back of the
Liouville forms in $\Lambda^m_2(J^1\pi)$, which we call the \textsl{symmetrized Liouville forms}.

Finally, let us consider the quotient bundle $J^2\pi^\ddagger = J^2\pi^\dagger / \Lambda^m_1(J^1\pi)$,
which is called the \textsl{restricted $2$-symmetric multimomentum bundle}. This bundle is  endowed
with a natural quotient map, $\mu \colon J^2\pi^\dagger \to J^2\pi^\ddagger$, and the natural
projections $\pi_{J^1\pi}^\ddagger \colon J^2\pi^\ddagger \to J^1\pi$ and
$\bar{\pi}_{J^1\pi}^\ddagger \colon J^2\pi^\ddagger \to M$.
Observe that $\dim J^{2}\pi^\ddagger = \dim J^{2}\pi^\dagger - 1$.

\section{Lagrangian-Hamiltonian unified formalism}
\label{laghamunif}

(See \cite{art:Prieto_Roman14} for details).
Let $\pi \colon E \to M$ be the configuration bundle of a second-order field theory, where $M$ is
an orientable $m$-dimensional manifold with volume form $\eta \in \df^{m}(M)$, and
$\dim E = m+n$. Let $\Lag \in \df^{m}(J^2\pi)$ be a second-order Lagrangian density for this field
theory. The \textsl{$2$-symmetric jet-multimomentum bundles} are
$$
\W = J^{3}\pi \times_{J^{1}\pi} J^2\pi^\dagger \quad ; \quad 
\W_r = J^{3}\pi \times_{J^{1}\pi} \, J^{2}\pi^\ddagger \, .
$$

These bundles are endowed with the canonical projections
$\rho_1^r \colon \W_r \to J^{3}\pi$, $\rho_2 \colon \W \to J^{2}\pi^\dagger$,
$\rho^r_2 \colon \W_r \to J^{2}\pi^\ddagger$, and
$\rho_M^r \colon \W_r \to M$. In addition, the natural quotient map
$\mu \colon J^2\pi^\dagger \to J^2\pi^\ddagger$ induces a natural submersion $\mu_\W \colon \W \to \W_r$.

Using the canonical structures in $\W$ and $\W_r$, we define a \textsl{Hamiltonian section}
$\hat{h} \in \Gamma(\mu_\W)$, which is specified by giving a local \textsl{Hamiltonian function}
$\hat{H} \in C^\infty(\W_r)$.
Then we define the forms $\Theta_r = (\rho_2 \circ \hat{h})^*\Theta \in \df^{m}(\W_r)$ and
$\Omega_r = -\d\Theta_r \in \df^{m+1}(\W_r)$.
Finally, $\psi \in \Gamma(\rho_M^r)$ is \textsl{holonomic} in $\W_r$ if
$\rho_1^r \circ \psi \in \Gamma(\bar{\pi}^{3})$ is holonomic in $J^{3}\pi$.

The \textsl{Lagrangian-Hamiltonian problem for sections} associated with the system $(\W_r,\Omega_r)$
consists in finding holonomic sections $\psi \in \Gamma(\rho_M^r)$ satisfying
\begin{equation}\label{eqn:UnifFieldEqSect}
\psi^*\inn(X)\Omega_r = 0 \, , \quad \mbox{for every } X \in \vf(\W_r) \, .
\end{equation}

\begin{proposition}\label{prop:GraphLegMapSect}
A section $\psi \in \Gamma(\rho_M^r)$ solution to the equation \eqref{eqn:UnifFieldEqSect} takes values in
a $n(m + m(m+1)/2)$-codimensional submanifold $j_\Lag \colon \W_\Lag \hookrightarrow \W_r$ which
is identified with the graph of a bundle map $\Leg \colon J^3\pi \to J^{2}\pi^\ddagger$ over
$J^1\pi$ defined locally by
$$
\Leg^*p^i_\alpha = \derpar{\hat{L}}{u_i^\alpha}
- \sum_{j=1}^{m}\frac{1}{n(ij)} \frac{d}{dx^j}\left( \derpar{\hat{L}}{u_{1_i+1_j}^\alpha} \right) \quad ; \quad
\Leg^*p^I_\alpha = \derpar{\hat{L}}{u_I^\alpha} \, .
$$
\end{proposition}

The map $\Leg$ is the \textsl{restricted Legendre map} associated with $\Lag$, and it can be
extended to a map $\widetilde{\Leg} \colon J^3\pi \to J^2\pi^\dagger$, which is called the
\textsl{extended Legendre map}.

\section{Variational Principle for the unified formalism}
\label{sec:VariationalPrincipleUnified}

If $\Gamma(\rho_M^r)$ is the set of sections of $\rho_M^r$, we consider the following
functional (where the convergence of the integral is assumed)
\begin{equation*}
\begin{array}{rcl}
\mathbf{LH} \colon \Gamma(\rho_M^r) & \longrightarrow & \R \\
\psi & \longmapsto & \displaystyle \int_M \psi^*\Theta_r
\end{array}
\end{equation*}

\begin{definition}[Generalized Variational Principle]
The \textnormal{Lagrangian-Hamiltonian variational problem} for the field theory $(\W_r,\Omega_r)$
is the search for the critical holonomic sections of the functional $\mathbf{LH}$ with respect to
the variations of $\psi$ given by $\psi_t = \sigma_t \circ \psi$, where $\left\{ \sigma_t \right\}$
is a local one-parameter group of any compact-supported $\rho_M^r$-vertical vector field $Z$ in
$\W_r$, that is,
\begin{equation*}
\restric{\frac{d}{dt}}{t=0}\int_M \psi_t^*\Theta_r = 0 \, .
\end{equation*}
\end{definition}

\begin{theorem}
\label{thm:EquivalenceVariationalSectionsUnified}
A holonomic section $\psi \in \Gamma(\rho_M^r)$ is a solution to the Lagrangian-Hamiltonian
variational problem if, and only if, it is a solution to equation
\eqref{eqn:UnifFieldEqSect}.
\end{theorem}
\begin{proof}
This proof follows the patterns in \cite{art:Echeverria_DeLeon_Munoz_Roman07} (see also
\cite{proc:Garcia_Munoz83}). Let $Z \in \vf^{V(\rho_M^r)}(\W_r)$ be a compact-supported vector
field, and $V \subset M$ an open set such that $\partial V$ is a $(m-1)$-dimensional manifold
and $\rho_M^r(\textnormal{supp}(Z)) \subset V$. Then,{\small
\begin{align*}
\restric{\frac{d}{dt}}{t=0}& \int_M \psi^*_t\Theta_r
= \restric{\frac{d}{dt}}{t=0} \int_V \psi^*_t\Theta_r
= \restric{\frac{d}{dt}}{t=0} \int_V \psi^*\sigma_t^*\Theta_r
= \int_V\psi^*\left( \lim_{t \to 0} \frac{\sigma_t^*\Theta_r - \Theta_r}{t} \right) \\
&= \int_V\psi^*\Lie(Z)\Theta_r
= \int_V \psi^*(\inn(Z)\d \Theta_r + \d\inn(Z)\Theta_r)
= \int_V \psi^*(-\inn(Z)\Omega_r + \d\inn(Z)\Theta_r)\\
&= - \int_V \psi^*\inn(Z)\Omega_r + \int_V \d(\psi^*\inn(Z)\Theta_r) 
= - \int_V \psi^*\inn(Z)\Omega_r + \int_{\partial V}\psi^*\inn(Z)\Theta_r \\
&= - \int_V\psi^*\inn(Z)\Omega_r \, ,
\end{align*}}
\normalsize as a consequence of Stoke's theorem and the assumptions made on the supports of the vertical
vector fields. Thus, by the fundamental theorem of the variational calculus, we conclude
$$
\restric{\frac{d}{dt}}{t=0} \int_M \psi_t^*\Theta_r = 0 \quad
\Longleftrightarrow \quad \psi^*\inn(Z)\Omega_r = 0 \, ,
$$
for every compact-supported $Z \in \vf^{V(\rho_M^r)}(\W_r)$. However, since the compact-supported
vector fields generate locally the $C^\infty(\W_r)$-module of vector fields in $\W_r$, it follows
that the last equality holds for every $\rho_M^r$-vertical vector field $Z$ in $\W_r$. Now, for every
$w \in \textnormal{Im}\psi$, we have a canonical splitting of the tangent space of $\W_r$ at $w$ in
a $\rho_M^r$-vertical subspace and a $\rho_M^r$-horizontal subspace,
$$
\Tan_w\W_r = V_w(\rho_M^r) \oplus \Tan_w(\textnormal{Im}\psi) \, .
$$
Thus, if $Y \in \vf(\W_r)$, then
$$
Y_w = (Y_w - \Tan_w(\psi \circ \rho_M^r)(Y_w)) + 
\Tan_w(\psi \circ \rho_M^r)(Y_w) \equiv Y_w^V + Y_w^{\psi} \, ,
$$
with $Y_w^V \in V_w(\rho_M^r)$ and $Y_w^{\psi} \in \Tan_w(\textnormal{Im}\psi)$. Therefore
$$
\psi^*\inn(Y)\Omega_r= \psi^*\inn(Y^V)\Omega_r + \psi^*\inn(Y^{\psi})\Omega_r = 
\psi^*\inn(Y^{\psi})\Omega_r \, ,
$$
since $\psi^*\inn(Y^V)\Omega_r = 0$, by the conclusion in the above paragraph. Now, as
$Y^{\psi}_w \in \Tan_w(\textnormal{Im}\psi)$ for every $w \in {\rm Im}\psi$, then the vector
field $Y^{\psi}$ is tangent to $\textnormal{Im}\psi$, and hence there exists a vector field
$X \in \vf(M)$ such that $X$ is $\psi$-related with $Y^{\psi}$; that is,
$\psi_*X = \restric{Y^{\psi}}{\textnormal{Im}\psi}$. Then
$\psi^*\inn(Y^{\psi})\Omega_r = \inn(X)\psi^*\Omega_r$. However, as $\dim\textnormal{Im}\psi = \dim M = m$ and
$\Omega_r$ is a $(m+1)$-form, we obtain that $\psi^*\inn(Y^{\psi})\Omega_r = 0$.
Hence, we conclude that $\psi^*\inn(Y)\Omega_r = 0$ for every $Y \in \vf(\W_r)$.

Taking into account the reasoning of the first paragraph, the converse is obvious 
since the condition $\psi^*\inn(Y)\Omega_r = 0$, for every $Y \in \vf(\W_r)$,
holds, in particular, for every $Z \in \vf^{V(\rho_\R^r)}(\W_r)$.
\end{proof}

\section{Lagrangian variational problem}
\label{Lagvp}

Consider the submanifold $j_\Lag\colon \W_\Lag \hookrightarrow \W_r$.
Since $\W_\Lag$ is the graph of the restricted Legendre map, the map
$\rho_1^\Lag = \rho_1^r \circ j_\Lag \colon \W_\Lag \to J^{3}\pi$ is a diffeomorphism.
Then we can define the \textsl{Poincar\'{e}-Cartan $m$-form} as
$\Theta_{\Lag} = (j_\Lag\circ(\rho^\Lag_1)^{-1})^* \Theta_r \in \df^{m}(J^3\pi)$.
This form coincides with the usual Poincar\'{e}-Cartan $m$-form derived in
\cite{art:Saunders87,art:Saunders_Crampin90}.

Given the Lagrangian field theory $(J^{3}\pi,\Omega_\Lag)$, consider the following functional
\begin{equation*} \label{eqn:DefnFunctionalVariationalLagrangian}
\begin{array}{rcl}
\mathbf{L} \colon \Gamma(\pi) & \longrightarrow & \R \\
\phi & \longmapsto & \displaystyle \int_M (j^{3}\phi)^*\Theta_\Lag
\end{array}
\end{equation*}

\begin{definition}[Generalized Hamilton Variational Principle]
The \textnormal{Lagrangian variational problem} (or \textnormal{Hamilton variational problem})
for the second-order Lagrangian field theory $(J^{3}\pi,\Omega_\Lag)$ is the search for the critical
sections of the functional $\mathbf{L}$ with respect to the variations of $\phi$
given by $\phi_t = \sigma_t \circ \phi$, where $\left\{ \sigma_t \right\}$ is a local one-parameter
group of any compact-supported $Z \in \vf^{V(\pi)}(E)$; that is,
\begin{equation*}
\restric{\frac{d}{dt}}{t=0}\int_M (j^{3}\phi_t)^*\Theta_\Lag = 0 \, .
\end{equation*}
\end{definition}

\begin{theorem} \label{thm:LagrangianVariational}
Let $\psi \in \Gamma(\rho_M^r)$ be a holonomic section which is critical for the functional
$\mathbf{LH}$. Then, $\phi = \pi^{3} \circ \rho_1^r \circ \psi \in \Gamma(\pi)$ is critical for
the functional $\mathbf{L}$.

\noindent Conversely, if $\phi \in \Gamma(\pi)$ is a critical section for the functional $\mathbf{L}$,
then the section $\psi = j_\Lag \circ (\rho_1^\Lag)^{-1} \circ j^3\phi \in \Gamma(\rho_M^r)$
is holonomic and it is critical for the functional $\mathbf{LH}$.
\end{theorem}
\begin{proof}
The proof follows the same patterns as in Theorem \ref{thm:EquivalenceVariationalSectionsUnified}.
The same reasoning also proves the converse.

\end{proof}

\section{Hamiltonian variational problem}
\label{Hamvp}

Let $\widetilde{\P} = \textnormal{Im}(\widetilde{\Leg}) \stackrel {\tilde{\jmath}}{\hookrightarrow} J^2\pi^\dagger$
and $\P = \textnormal{Im}(\Leg) \stackrel{\jmath}{\hookrightarrow}J^{2}\pi^\ddagger$ the image of
the extended and restricted Legendre maps, respectively; $\bar{\pi}_\P \colon \P \to M$ the
natural projection, and $\Leg_o \colon J^3\pi \to \P$ the map defined by $\Leg = \jmath \circ \Leg_o$.

A Lagrangian density $\Lag \in \df^{m}(J^{2}\pi)$ is \textsl{almost-regular} if
(i) $\P$ is a closed submanifold of $J^{2}\pi^\ddagger$,
(ii) $\Leg$ is a submersion onto its image, and
(iii) for every $j^3_x\phi \in J^{3}\pi$, the fibers $\Leg^{-1}(\Leg(j^3_x\phi))$ are connected
submanifolds of $J^{3}\pi$.

The Hamiltonian section $\hat{h} \in \Gamma(\mu_\W)$ induces a
Hamiltonian section $h \in \Gamma(\mu)$ defined by
$\rho_2\circ \hat{h} = h \circ \rho^r_2$.
Then, we define the \textsl{Hamilton-Cartan $m$-form} in $\P$ as
$\Theta_h = (h \circ \jmath)^*\Theta_{1}^s \in \df^{m}(\P)$.
Observe that $\Leg_o^*\Theta_h = \Theta_\Lag$.

In what follows, we consider that the Lagrangian density $\Lag \in \df^{m}(J^{2}\pi)$ is, at least,
almost-regular. Given the Hamiltonian field theory $(\P,\Omega_h)$, let $\Gamma(\bar{\pi}_\P)$ be
the set of sections of $\bar{\pi}_\P$. Consider the following functional
\begin{equation*}
\begin{array}{rcl}
\mathbf{H} \colon \Gamma(\bar{\pi}_\P) & \longrightarrow & \R \\
\psi_h & \longmapsto & \displaystyle \int_M \psi_h^*\Theta_{\P}
\end{array}
\end{equation*}

\begin{definition}[Generalized Hamilton-Jacobi Variational Principle]
The \textnormal{Hamiltonian variational problem} (or \textnormal{Hamilton-Jacobi variational problem})
for the second-order Hamiltonian field theory $(\P,\Omega_{h})$ is the search for the critical
sections of the functional $\mathbf{H}$ with respect to the variations
of $\psi_h$ given by $(\psi_h)_t = \sigma_t \circ \psi_h$, where $\left\{ \sigma_t \right\}$ is a
local one-parameter group of any compact-supported $Z \in \vf^{V(\bar{\pi}_\P)}(\P)$,
\begin{equation*}
\restric{\frac{d}{dt}}{t=0}\int_M (\psi_h)_t^*\Theta_{h} = 0 \, .
\end{equation*}
\end{definition}

\begin{theorem} \label{thm:HamiltonianVariational}
Let $\psi \in \Gamma(\rho_M^r)$ be a critical section of the functional $\mathbf{LH}$. Then, the
section $\psi_h = \Leg_o \circ \rho_1^r \circ \psi \in \Gamma(\bar{\pi}_\P)$ is a critical section
of the functional $\mathbf{H}$.

\noindent Conversely, if $\psi_h \in \Gamma(\bar{\pi}_\P)$ is a critical section of the functional $\mathbf{H}$,
then the section $\psi = j_\Lag \circ (\rho_1^\Lag)^{-1} \circ \gamma \circ \psi_h \in \Gamma(\rho_M^r)$
is a critical section of the functional $\mathbf{LH}$, where $\gamma \in \Gamma_\P(\Leg_o)$ is a local
section of $\Leg_o$.
\end{theorem}
\begin{proof}
The proof follows the same patterns as in Theorem \ref{thm:EquivalenceVariationalSectionsUnified}.
The same reasoning also proves the converse, bearing in mind that $\gamma \in \Gamma_\P(\Leg_o)$
is a local section.

\end{proof}

\section{The higher-order case}

As stated in \cite{art:Prieto_Roman14}, this formulation fails when we try to generalize it to
a classical field theory of order greater or equal than $3$. The main obstruction to do so is
the relation among the multimomentum coordinates used to define the submanifold $J^2\pi^\dagger$,
$p^{ij}_\alpha = p_{\alpha}^{ji}$ for every $1 \leqslant i,j \leqslant m$ and every
$1 \leqslant \alpha \leqslant n$. Although this ``symmetry'' relation on the multimomentum
coordinates can indeed be generalized to higher-order field theories, it only holds for the
highest-order multimomenta. That is, this relation on the multimomenta is not invariant
under change of coordinates for lower orders, and hence we do not obtain a submanifold of
$\Lambda^m_2(J^{k-1}\pi)$.

\section*{Acknowledgments}

We acknowledge the financial support of the  \textsl{MICINN}, projects MTM2011-22585 and
MTM2011-15725-E. P.D. Prieto-Mart\'{\i}nez thanks the UPC for a Ph.D grant.

\end{document}